\documentclass{article}
\usepackage{amssymb,mathtools,import,geometry,graphicx,amsthm}

\newcommand{\N}{\ensuremath{\mathbb{N}} }

\newcommand{\fbins}{\ensuremath{\{0,1\}^*}}
\newcommand{\bin}{\ensuremath{\{0,1\}}}
\newcommand{\infbins}{\ensuremath{\{0,1\}^{\omega}}}
\newcommand{\PPMk}[1]{\ensuremath{\mathrm{PPM}_{#1}}}
\newcommand{\PPMun}{\ensuremath{\mathrm{PPM}^*}}
\newcommand{\LZ}{\ensuremath{\mathrm{LZ}}}
\newcommand{\thh}{\ensuremath{\textrm{th}}}

\newcommand{\occ}[2]{\ensuremath{\text{occ}(#1,#2)}}

\newcommand{\occb}[2]{\ensuremath{\text{occ$_b$}(#1,#2)}}

\theoremstyle{plain}
\newtheorem{theorem}{Theorem}[section]
\newtheorem{corollary}[theorem]{Corollary}
\newtheorem{lemma}[theorem]{Lemma}
\theoremstyle{definition}

\newtheorem{remark}[theorem]{Remark}
\newtheorem{claim}[theorem]{Claim}

\begin{document}

\title{A Normal Sequence Compressed by \PPMun \, but not by Lempel-Ziv 78}

\author{Liam Jordon\thanks{Supported by a postgraduate scholarship from the Irish Research Council.}\\
\textrm{liam.jordon@mu.ie} \\
\and 
Philippe Moser \\
\textrm{pmoser@cs.nuim.ie}}

\date{%
    Dept. of Computer Science, Maynooth University, Maynooth, Co. Kildare, Ireland.\\[2ex]%
}

%\affil{Department of Computer Science, Maynooth University, Maynooth, Co. Kildare, Ireland}

\maketitle

\begin{abstract}
In this paper we compare the difference in performance of two of the Prediction by Partial Matching (PPM) family of compressors (\PPMun and the original Bounded PPM algorithm) and the Lempel-Ziv 78 (LZ) algorithm. We construct an infinite binary sequence whose worst-case compression ratio for PPM* is $0$, while Bounded PPM's and LZ's best-case compression ratios are at least $1/2$ and $1$ respectively. This sequence is an enumeration of all binary strings in order of length, i.e. all strings of length $1$ followed by all strings of length $2$ and so on. It is therefore normal, and is built using repetitions of \emph{de Bruijn} strings of increasing order. 
\end{abstract}


\section{Introduction}

A \textit{normal} number in base $b$, as defined by Borel \cite{borelNormal}, is a real number whose infinite decimal expansion in that base is such that for all block lengths $n$, every string of digits in base $b$ of length $n$ occur as a substring in the decimal expansion with limiting frequency $\frac{1}{b^n}$. In this paper we restrict ourselves to examining normal binary sequences, i.e. normal numbers in base $2$.

A common question studied about normal sequences is whether or not they are compressible by certain families of compressors. Results by Schnorr and Stimm \cite{DBLP:journals/acta/SchnorrS72} and Dai, Lathrop, Lutz and Mayordomo \cite{b.lutz.finite-state.dimension} demonstrate that lossless finite-state transducers (FSTs) cannot compress normal sequences. Becher, Carton and Heiber \cite{DBLP:journals/jcss/BecherCH15} explore what happens to the compressibilty of normal sequences in various scenarios such as when the FST has access to one or more counters or a stack, and what happens when the transducer is not required to run in real-time nor be deterministic. Carton and Heiber \cite{DBLP:journals/iandc/CartonH15} show that deterministic and non-deterministic two-way FSTs cannot compress normal sequences. Among other compression algorithms, Lathrop and Strauss \cite{maliciousLZ} have shown that there exists a normal sequence such that the Lempel-Ziv 78 (LZ) algorithm can compress. 

In this paper we focus on the performance of the Prediction by Partial Matching (PPM) compression algorithm which was introduced by Cleary and Witten \cite{DBLP:journals/tcom/ClearyW84}. PPM works by building an adaptive statistical model of its input as it reads each character. The model keeps track of previously seen substrings in the input, known as \emph{contexts}, and the characters that follow them. When predicting the next character, the encoder begins by identifying the \emph{relevant} contexts currently in the model. These relevant contexts refer to suffixes of the of the already encoded part of the input that have been stored in the model. The next character is then encoded based on its frequency counts in the relevant contexts. The model is updated after each character is encoded. This involves updating the frequency counts of the seen character in the relevant contexts and, if needed, adding new contexts in the model. These prediction probabilities for each character encodes the sequence via arithmetic encoding \cite{WittenACEncoding}.

In the original PPM (Bounded PPM), prior to encoding the input, a value $k \in \N$ must be provided to the encoder which sets the maximum length of a context the model can store. Studies have gone into identifying which value for $k$ achieves the best compression. One may think the larger the $k$, the better the compression. However, increasing $k$ above $5$ does not generally improve compression \cite{DBLP:journals/tcom/ClearyW84}. Over a decade later, a new version of PPM was introduced, called PPM* \cite{PPMstar}. This version of the algorithm sets no upper bound on the length of contexts the model can keep track of. 

Inspired by Mayordomo, Moser and Perifel \cite{DBLP:journals/mst/MayordomoMP11} which compares the best-case and worse-case compression ratio of various compression algorithms on certain sequences, in this paper we construct a normal sequence $S$ and compare how it is compressed by $\PPMun$, Bounded PPM and LZ. $\PPMun$  can compress $S$ with a worst-case compression ratio of $0$. We also show that no matter what upper bound for $k$ chosen, Bounded PPM's best-case compression ratio is at least $1/2$. Also, LZ has a best-case compression ratio of $1$ on $S$, i.e. $S$ cannot be compressed by LZ.

$S$ is constructed such that it is an enumeration of all binary strings in order of length i.e. all strings of length $1$ followed by all strings of length $2$ and so on. For instance, $0100011011$ is an enumeration of all strings up to length $2$. Such sequences cannot be compressed by LZ, which in turn means they cannot be compressed by any FST \cite{b.lempel-ziv}. Thus $S$ is normal. This enumeration is achieved via repetitions of \emph{de Bruijn} strings which \PPMun \, can exploit to compress $S$.

Some proofs are omitted from the main body of the paper due to space constraints. These are all contained in the appendix provided.

\section{Preliminaries}

$\N$ denotes the set of non-negative integers. A \emph{finite binary} string is an element of $\fbins$. A \emph{binary sequence} is an element of $\infbins$. The length of a string $x$ is denoted by $|x|$. $\lambda$ denotes the empty string, i.e. the string of length $0$. For all $n \in \N$, $\bin^n$ denotes the set of binary strings of length $n$. For a string (or sequence) $x$ and $i,j, \in \N$ with $i \leq j$, $x[i..j]$ denotes the $i^{\thh}$ through $j^{\thh}$ bits of $x$ with the convention that if $j < i$ then $x[i..j] = \lambda$. For a string $x$ and string (sequence) $y$, $xy$ (occasionally denoted by $x \cdot y$) denotes the string (sequence) of $x$ concatenated with $y$. For a string $x$ and $n \in \N$, $x^n$ denotes $x$ concatenated with itself $n$  times. For strings $x,y$ and string (sequence) $z$, if $w = xyz$, we say $y$ is a substring of $w$, $x$ is a prefix of $w$ (sometimes denoted by $x \sqsubseteq w$), and if $z$ is a string, then $z$ is a suffix of $w$.   For a sequence $S$ and $n \in \N$, $S \upharpoonright n$ denotes the prefix of $S$ of length $n$, i.e. $S \upharpoonright n = S[0..n-1].$ The \emph{lexicographic} ordering of $\fbins$ is defined by saying for two strings $x,y$, $x$ is less than $y$ if either $|x| < |y|$ or else $|x| = |y|$ with $x[n..n] = 0$ and $y[n..n] = 1$ for the least $n$ such that $x[n] \neq y[n]$.

Given a sequence $S$ and a function $T: \fbins \rightarrow \fbins$, the \emph{best-case} and  \emph{worst-case compression ratios} of $T$ on $S$ are given by  
$$\rho_{T}(S) =  \liminf\limits_{n \to \infty} \frac{|T(S \upharpoonright n)|}{n} \text{ and, } R_{T}(S) =  \limsup\limits_{n \to \infty} \frac{|T(S \upharpoonright n)|}{n}$$ respectively.

Given strings $x,w$ we use the following notation to count the number of times $w$ occurs as a substring in $x$.
\begin{enumerate}
    \item The number of occurrences of $w$ as a substring of $x$ is given by
    $$\occ{w}{x} = |\{u \in \fbins : uw \sqsubseteq x \}|.$$
    
    \item The block number of occurrences of $w$ as a substring of $x$ is given by
    $$\occb{w}{x} = |\{i : x[i + i+|w|-1] = w, i \equiv 0 \mod |w| \}|$$
\end{enumerate}

A sequence $S$ is said to be \emph{normal}, as defined by Borel \cite{borelNormal} if for all $w \in \fbins$

$$\lim\limits_{n \rightarrow \infty}\frac{\occ{w}{S \upharpoonright n}}{n} = 2^{-|w|}.$$

We say that a sequence $S$ is an \emph{enumeration of all strings}, we mean that $S$ can be broken into substrings $S=S_1S_2S_3\ldots,$ such that for each $n$, $S_n$ is a concatenation of all strings of length $n$ with each string occurring once. That is, for all $w \in \bin^i, \occb{w}{S_i} = 1.$ Note that $|S_n| = n(2^n).$

\section{Description of the PPM Algorithms}
 \label{ppm section}
 Before we begin, we note that implementations of the PPM algorithm family implement what is known as the \emph{exclusion principal} to achieve better compression ratios. We ignore this in our implementation for simplicity as even without this, the sequence we later build achieves a compression ratio of $0$ via $\PPMun$.

 \subsection{Bounded PPM}
 
 In the original presentation of PPM in 1984 \cite{DBLP:journals/tcom/ClearyW84}, a bounded version is introduced. Prior to encoding the input sequence, a value $k \in \N$ must be provided to the encoder which sets the maximum \emph{context} length the model keeps track of. As such, we refer to this version as Bounded PPM and denote Bounded PPM with bound $k$ as $\PPMk{k}.$ By context, we mean previously seen substrings of the input stream contained in the model. For each context, the model records what characters have followed the context in the input stream, and the frequency each character has occurred. These frequencies are used to build \emph{prediction probabilities} that the encoder uses to encode the rest of the input stream. When reading the next bit of the input stream, the encoder examines the longest \emph{relevant} context each time and encodes the current character based on its current prediction probability in that context. By relevant context, we mean suffixes of the input stream already read by the encoder that are contained in the model. The longest relevant context available is referred to as the \emph{current} context as it is the one the model uses to first encode the next character seen. Once encoded, the model is updated to include new contexts if necessary, and to update the prediction probabilities of the relevant contexts to reflect the character that has just been read.
 
 A problem occurs if the character being encoded has never occurred previously in the current context. When this occurs, an \emph{escape} symbol (denoted by \$) is transmitted and the next shortest relevant context becomes the new current context. If the character has not been seen before even when the current context is $\lambda$, that is, the context where none of the previous bits are used to predict the next character, an escape is outputted and the character is assigned the prediction probability from the order-$(-1)$ table. By convention, this table contains all characters in the alphabet being used and assigns each character equal probability.

A common question is what probabilities are assigned to these escape symbols. This paper uses \textit{Method C} proposed by Moffat \cite{moffatMethodC}. Here, the escape symbol is given a frequency equal to the number of distinct characters predicted in the context so far. Hence in our case it will always have a count of $1$ or $2$.

For instance in Table \ref{tab:bd}, the model for the string $0100110110$ with bound $k=3$ is found. In the context $01$, the escape symbol $\$$ has count $2$ as both $0$ and $1$ have been seen, while $\$$ in $101$ has count $1$ as only $1$ has been seen. A context is said to be \emph{deterministic} if it has an escape count frequency of 1.

For example, suppose $0$ is the next character to be encoded after input stream $0100110110$ by $\PPMk{3}$, whose model is shown in Table \ref{tab:bd}. The relevant contexts are $110,10,0$ and $\lambda$. The longest relevant context is $110$. The encoder escapes to the shorter context $10$ since $0$ is not seen in context $110$ and is encoded by the prediction probability $\frac{1}{2}$. From $10$, $0$ is encoded with probability $\frac{1}{4}$. The frequency counts of $0$ will be updated in the $10$, $0$ and empty contexts. Also, $0$ will be added as a prediction to context $110$. Following this, if the next character to be encoded was another $0$ the model would start in context $100$, and since $0$ is not predicted here, it would transmit an escape symbol with probability $\frac{1}{2}$ and then examine the next longest context $00$ and proceed as necessary. If there was another bit $b$ in the input stream after this, as $000$ would be the  current suffix of the input stream but no context for $000$ exists yet as it has never been seen before, the encoder would begin in context $00$ and proceed as before, and a context for $000$ would be created predicting the character $b$ when the model updates.

\begin{table}[t]
\centering
\begin{tabular}{cccclcccclcccc}
\hline
ctxt & pred & cnt & pb &  & ctxt & pred & cnt & pb &  & ctxt & pred & cnt & pb \\ \hline

\multicolumn{4}{c}{Order $k = 3$} &  & \multicolumn{4}{c}{Order $k = 2$} &  & \multicolumn{4}{c}{Order $k = 1$} \\[1mm]
001 & 1 & 1 & $\frac{1}{2}$ &  & 00 & 1 & 1 & $\frac{1}{2}$ &  & 0 & 0 & 1 & $\frac{1}{6}$ \\[1mm]
 & \textit{\$} & 1 & $\frac{1}{2}$ &  &  & \textit{\$} & 1 & $\frac{1}{2}$ &  &  & 1 & 3 & $\frac{1}{2}$\\[1mm]
010 & 0 & 1 & $\frac{1}{2}$ &  & 01 & 0 & 1 & $\frac{1}{5}$ &  &  & \textit{\$} & 2 & $\frac{1}{3}$ \\[1mm]
 & \textit{\$} & 1 & $\frac{1}{2}$ &  &  & 1 & 2 & $\frac{2}{5}$ &  & 1 & 0 & 3 & $\frac{3}{7}$ \\[1mm]
011 & 0 & 2 & $\frac{2}{3}$ &  &  & \textit{\$} & 2 & $\frac{2}{5}$ &  &  & 1 & 2 & $\frac{2}{7}$ \\[1mm]
 & \textit{\$} & 1 & $\frac{1}{3}$ &  & 10 & 0 & 1 & $\frac{1}{4}$ &  &  & \textit{\$} & 2 & $\frac{2}{7}$ \\[1mm]
100 & 1 & 1 & $\frac{1}{2}$ &  &  & 1 & 1 & $\frac{1}{4}$ &  & \multicolumn{4}{c}{Order $k = 0$} \\[1mm]
 & \textit{\$} & 1 & $\frac{1}{2}$ &  &  & \textit{\$} & 2 & $\frac{1}{2}$ &  & $\lambda$ & 0 & 5 & $\frac{5}{12}$ \\[1mm]
101 & 1 & 1 & $\frac{1}{2}$ &  & 11 & 0 & 2 & $\frac{2}{3}$ &  &  & 1 & 5 & $\frac{5}{12}$ \\[1mm]
 & \textit{\$} & 1 & $\frac{1}{2}$ &  &  & \textit{\$} & 1 & $\frac{1}{3}$ &  &  & \textit{\$} & 2 & $\frac{1}{6}$ \\[1mm]
110 & 1 & 1 & $\frac{1}{2}$ &  &  &  &  &  &  & \multicolumn{4}{c}{Order $k = -1$} \\[1mm]
 & \textit{\$} & 1 & $\frac{1}{2}$ &  &  &  &  &  &  &  & 0 & 1 & $\frac{1}{2}$ \\[1mm]
 &  &  &  &  &  &  &  &  &  &  & 1 & 1 & $\frac{1}{2}$ \\[1mm]
\end{tabular}
\caption{$\PPMk{3}$ model for the input $0100110110$}
\label{tab:bd}
\end{table}

%REMOVED PPM TABLE

\subsection{$\PPMun$}

$\PPMun$ encodes its input very similarly to Bounded PPM in that it builds a model of contexts of the sequence, continuously updates the model, and encodes each character it sees based on its frequency probability of the current context. However there are some key differences. As there is no upper bound on the max context length stored in the model, instead of building a context for every substring seen in the input, a context is only extended until it is unique. Suppose that of the input stream to $\PPMun$, the prefix $x$ has been read so far. This means that for a string $w \in \fbins$, if $\occ{w}{x} \geq 2$, the context $wb$ must be built in the model for each $b$ that follows $w$ in $x$. When examining all relevant contexts to choose the be the first current context to encode the next bit, unlike Bounded PPM which chooses the longest context, $\PPMun$ chooses the shortest \emph{deterministic} context. Here, a context is said to be \emph{deterministic} if it has an escape count frequency of $1$. If no such context exists, the longest is chosen as with Bounded PPM. We also use the \emph{Method C} approach to computing escape probabilities for $\PPMun$.

The following is a full example of a model being updated. Suppose an input stream of $s = 0100110110$ has already been read. The model for this is seen in Table \ref{tab:ubd}. Say the next bit read is a $0$. The relevant contexts are the empty context, $0,10,110$ and $0110$. The shortest deterministic context is $110$. It does not predict a $0$ so an escape is transmitted with probability $\frac{1}{2}$ and then $0$ is transmitted from the context $10$ with probability $\frac{1}{4}$. The model is then updated as follows. The empty context, $0$ and $10$ all predict a $0$ so their counts are updated. $110$ and $0110$ don't predict a $0$, so it is added as a prediction. Furthermore, the substrings $00$ and $100$ are not unique in $s0$ while they were in just $s$. That is $\occ{00}{s0} \neq 1, \occ{100}{s0} \neq 1$ while $\occ{00}{s} = \occ{100}{s} = 1.$ These contexts must be extended to create new contexts $001$ and $1001$. This is because $1$ is what follows $00$ and $100$ in $s$. These contexts both predict $1$. If another $0$ is read after $s0$, since both a $0$ and $1$ now have been seen to follow $110$ and $0110$, contexts for $1100$ and $01100$ will be created both predicting a $0$, since a context has to be made for each branching path of $110$ and $0110$  ($1101$ and $01101$ already exist).

\begin{table}[t]
\centering
\begin{tabular}{cccclcccclcccc}
\hline
ctxt & pred & \multicolumn{1}{c}{cnt} & pb &  & ctxt & pred & cnt & pb &  & ctxt & pred & cnt & pb \\ \hline
\multicolumn{4}{c}{Order $k = 5$} &  & 101 & 1 & 1 & $\frac{1}{2}$ &  & \multicolumn{4}{c}{Order $k = 1$} \\[1mm]
01101 & 1 & 1 & $\frac{1}{2}$ &  &  & \textit{\$} & 1 & $\frac{1}{2}$ &  & 0 & 0 & 1 & $\frac{1}{6}$ \\[1mm]
 & \textit{\$} & 1 & $\frac{1}{2}$ &  & 110 & 1 & 1 & $\frac{1}{2}$ &  &  & 1 & 3 & $\frac{1}{2}$ \\[1mm]
\multicolumn{4}{c}{Order $k = 4$} &  &  & \textit{\$} & 1 & $\frac{1}{2}$ &  &  & \textit{\$} & 2 & $\frac{1}{3}$ \\[1mm]
0110 & 1 & 1 & $\frac{1}{2}$ &  & \multicolumn{4}{c}{Order $k = 2$} &  & 1 & 0 & 3 & $\frac{3}{7}$ \\[1mm]
 & \textit{\$} & 1 & $\frac{1}{2}$ &  & 00 & 1 & 1 & $\frac{1}{2}$ &  &  & 1 & 2 & $\frac{2}{7}$ \\[1mm]
1101 & 1 & 1 & $\frac{1}{2}$ &  &  & \textit{\$} & 1 & $\frac{1}{2}$ &  &  & \textit{\$} & 2 & $\frac{2}{7}$ \\[1mm]
 & \textit{\$} & 1 & $\frac{1}{2}$ &  & 01 & 0 & 1 & $\frac{1}{5}$ &  & \multicolumn{4}{c}{Order $k = 0$} \\[1mm]
\multicolumn{4}{c}{Order $k = 3$} &  &  & 1 & 2 & $\frac{2}{5}$ &  & $\lambda$ & 0 & 5 & $\frac{5}{12}$ \\[1mm]
010 & 0 & 1 & $\frac{1}{2}$ &  &  & \textit{\$} & 2 & $\frac{2}{5}$ &  &  & 1 & 5 & $\frac{5}{12}$ \\[1mm]
 & \textit{\$} & 1 & $\frac{1}{2}$ &  & 10 & 0 & 1 & $\frac{1}{4}$ &  &  & \textit{\$} & 2 & $\frac{1}{6}$ \\[1mm]
011 & 0 & 2 & $\frac{2}{3}$ &  &  & 1 & 1 & $\frac{1}{4}$ &  & \multicolumn{4}{c}{Order $k = -1$} \\[1mm]
 & \textit{\$} & 1 & $\frac{1}{3}$ &  &  & \textit{\$} & 2 & $\frac{1}{2}$ &  &  & 0 & 1 & $\frac{1}{2}$ \\[1mm]
100 & 1 & 1 & $\frac{1}{2}$ &  & 11 & 0 & 2 & $\frac{2}{3}$ &  &  & 1 & 1 & $\frac{1}{2}$ \\[1mm]
 & \textit{\$} & 1 & $\frac{1}{2}$ &  &  & \textit{\$} & 1 & $\frac{1}{3}$ &  &  &  &  & 
\end{tabular}
\caption{\label{tab:ubd}PPM* model for the input 0100110110}
\label{tab: unb}
\end{table}

\subsection{Arithmetic Encoding}

The final output of the PPM encoder is a real number in the interval $[0,1)$ found via arithmetic encoding \cite{Bell:1990:TC:77753,WittenACEncoding}. The arithmetic encoder begins with the interval $[0,1)$. At each stage of encoding, the interval is split into subintervals of lengths corresponding to the probabilities of the current context being examined by the model. The subinterval corresponding to the character or escape symbol transmitted is then carried forward to the next stage.

Once the final character of the input string is encoded, a real number $c$ is transmitted such that $ c \in [a,b)$, where $[a,b)$ is the final interval and $c$ can be encoded in $-\lceil \log(|b - a|) \rceil$ bits. At most 1 bit of overhead is required. With $c$ and the length of the original sequence to be encoded, the decoder can find the original sequence. 

For simplicity, we assume the encoder can calculate the endpoints of the intervals with infinite precision and waits until the end to convert the fraction to its final form at the end of the encoding. In reality, a fixed finite limit precision is used by encoders to represent the intervals and their endpoints and a process known as \emph{renormalisation} occurs to prevent the intervals becoming too small for the encoder to handle.

\section{An Analysis of $\PPMun$}
\label{PPM Analysis}

In this section we build a normal sequence $S$ such that $R_{\PPMun}(S) = 0$. $S$ will be an enumeration of all binary strings and built via a concatenations of \emph{de Bruijn} strings. This ensures it is imcompressible by the Lempel-Ziv algorithm.

\subsection{\emph{de Bruijn} Strings}

Named after Nicolaas de Bruijn for his work from 1946\cite{debruijn1946combinatorial}, for $n \in \N$, a \emph{de Bruijn} string of order $n$ is a string of length $2^n$ that when viewed cyclically, contains all binary strings of length $n$ exactly once. That is, for a \emph{de Bruijn} string $x$ of order $n$, for all $w \in \bin^n$, $\occ{w}{x \cdot x[0.. n-2]} = 1.$ For example, $00011101$ is a \emph{de Bruijn} string of order $3$.

Henceforth, we use $db(n)$ to denote the least lexicographic \emph{de Bruijn} string of order $n$. Martin provided the following algorithm to build this string in 1934 \cite{martin1934}:

\begin{enumerate}
    \item Write the string $x = 1^{n-1}$.
    \item While possible, append a bit (with $0$ taking priority over $1$) to the end of $x$ so that substrings of length $n$ occur only once in $x$.
    \item When step 2 is no longer possible \footnote{Martin proves that this occurs when $|x| = 2^n + n -1$}, remove the prefix $1^{n-1}$ from $x$. The resulting string is $db(n)$. 
\end{enumerate}

Before we proceed we make note of the following properties of $db(n).$
\begin{remark}
For $n \geq 3$, 
\begin{enumerate}
    \item $db(n)[0..2n] = 0^n10^{n-2}11$,
    \item $db(n)[2^n-n-1.. 2^n-1] = 1^n$.
    \end{enumerate}
    \label{db facts}
\end{remark}

\begin{proof}[Proof: ]

From the construction, $db(n)$ must begin with $0^n$. This is followed by a $1$, otherwise the string $0^n$ would occur twice. The next $n-2$ bits are $0$s. This must be followed by a $1$. For otherwise, if $0^n10^{n-2}$ was followed by another $0$, this would result in either $0^n$ or $0^{n-1}1$ occurring twice depending on whether the next bit was a $0$ or a $1$. $0^n10^{n-2}1$ is then followed by a $1$ as otherwise $0^{n-2}10$ occurs twice.

$db(n)$ having suffix of $1^n$ is proven by Martin when showing when his algorithm terminates \cite{martin1934}. 
\end{proof}

%The proof is contained in the appendix.

We use the following notation for cyclic shifts of $db(n)$. For $0 \leq i < 2^n $, let $db_i(n)$ denote a left shift of $i$ bits of $db(n)$. That is, $db_i(n) = db(n)[i ..2^n-1] \cdot  db(n)[0.. i-1]$. We write $db(n)$ instead of $db_0(n)$ when no shift occurs. $db_i^j(n)$ denotes $db_i(n)$ concatenated with itself $j$ times.

\subsection{Construction and Properties of $S$}

The infinite binary sequence $S = S_1S_2S_3\ldots$ is built such that each $S_n$ is a concatenation of all strings of length $n$ and maximises repetitions to exploit PPM*. Maximising repetitions ensures deterministic contexts are repeatedly used to predict bits in the sequence, thus resulting in compression.

For every $n \in \N$, $n$ can be written in the form $n=2^st$, where $s,t \in \N \cup \{0\}$ and $t$  is odd. We set $S_n = B_{n,0}\cdot B_{n,1}\cdots B_{n,2^s - 1}$, where $B_{n,i} = db_{i}^t(n)$. Each $B_{n,i}$ is called the $i^\thh$ \emph{block} of $S_n$. Note that if $n$ is odd then $S_n = db^n(n),$ and if $n$ is a power of $2$ then $S_n = db(n)\cdot db_1(n)\cdots db_{n-1}(n).$ To help visualise this, table \ref{S6} is provided which shows how $S_6$ is built.

\begin{table}[t]
\centering
\begin{tabular}{l}
0000001000011000101000111001001011001101001111010101110110111111 \\
0000001000011000101000111001001011001101001111010101110110111111 \\
0000001000011000101000111001001011001101001111010101110110111111 \\
0000010000110001010001110010010110011010011110101011101101111110 \\
0000010000110001010001110010010110011010011110101011101101111110 \\
0000010000110001010001110010010110011010011110101011101101111110
\end{tabular}
\caption{\label{S6}To construct $S_6$, concatenate the six rows of this table. The first three rows are $db^3(6)$ while the second three rows are $db^3_1(6)$.}
\end{table}

The following lemma states that $S$ is in fact an enumeration of all binary strings, and hence normal. This is the property which later ensures that $S$ is Lempel-Ziv imcompressible. %The proof is contained in the appendix.

\begin{lemma}
\label{Enum Proof}
For each $n \in \N$, for $w \in \bin^n$, $\occb{w}{S_n} = 1.$
\end{lemma}

\begin{proof}[Proof: ]

Consider the cyclic group of order $2^n, C_{2^n} = \langle x\, | x^{2^n} = e\rangle$, where $e=x^0$ is the identity element and $x$ is the generator of the group. There exists a bijective mapping  $f:C_{2^n} \rightarrow \{0,1\}^n$ such that for $0 \leq a < 2^n$, $x^a$ is mapped to the substring of $db(n)$ of length $n$ beginning in position $a$. That is, $f(e) = db(n)[0.. n-1], f(x) = db(n)[1 .. n], \ldots ,f(x^{2^n - 1}) = db(n)[2^n - 1]\cdot db(n)[0 .. n-2].$

Let $s,t \in \N \cup \{0\}$ such that $t$  is odd and $n=2^st$. Consider the subgroup $\langle x^n\rangle$ of $C_{2^n}$. From group theory it follows that $$|\langle x^n\rangle | = \frac{2^n}{gcd(n,2^n)} = 2^{2^st - s} = 2^{n - s}.$$ So $$\langle x^n\rangle = \bigcup\limits_{i=0}^{2^{n-s}-1}\{x^{in \mod 2^n}\} = \{e,x^n,x^{2n},\ldots ,x^{(2^{n-s}-1)n \mod 2^n}\}.$$ Concatenating the result of applying $f$ to each element of $\langle x^n\rangle $ beginning with $e$ in the natural order gives the string $$\sigma = f(e)\cdot f(x^n)\cdot f(x^{2n})\cdots f(x^{(2^{n-s}-1)n \mod 2^n}).$$ $\sigma$ can be thought of as beginning with the string $0^n$, cycling through $db(n)$ in blocks of size $n$ until the block $1^n$ is seen. As $\frac{2^{n - s}n}{2^n} = t$, we have that $ \sigma = db^t(n) = B_0$.

As $\frac{|C_{2^n}|}{|\langle x^n\rangle |} = 2^s$, there are $2^s$ cosets of $\langle x^n\rangle$ in $C_{2^n}$. As each coset is disjoint, each represents a different set of $2^{n - s}$ strings of $\bin^n$. Specifically each coset represents $B_i = db_{i}^t(n)$. Therefore, for each $y \in \bin^n$, for some $i \in \{0,\ldots ,2^s-1\}$, $\occb{y}{B_i} = 1$ and $\occb{y}{B_j} = 0$ for each $j\neq i$. Thus $S_n$ is an enumeration of $\{0,1\}^n$.
 
\end{proof}

We proceed by examining some basic properties of each $S_n$ section of $S$ for $n$ large. Henceforth, we write $S \upharpoonright S_n$ to denote $S_1\cdots S_n$.

Suppose the encoder has already processed $S \upharpoonright S_{n-1}$, so the next bit to be processed is the first bit of $S_n$. While the encoder's model may contain contexts of length $n$ after processing $S \upharpoonright S_{n-1}$, the following lemma shows it will contain all possible contexts of length $n$ after reading the first $2^n + n$ bits of $S_n$. The idea is that in the first $2^n + n - 1$ bits of $S_n$, for each $x \in \bin ^{n-1}$, $x$ occurs at least twice, and $x0$ and $x1$ occur once. Hence a context for each branching path of $x$ must be created, i.e. contexts for $x0$ and $x1$.

\begin{lemma}
\label{context building}
Let $n\geq 2$ and suppose the encoder has already processed $S \upharpoonright S_{n-1}$. The encoder's model will contain contexts for all $w \in \bin^n$ once it has processed the $S[0..(2^n+n-1)]$, i.e. the first $2^n + n$ bits of $S_n$.
\end{lemma}

\begin{proof}[Proof: ][Lemma \ref{context building}]

Consider $x = S_n[0 .. 2^n + (n-2)] = db(n)\cdot 0^{n-1}$ (by Remark \ref{db facts}). By the definition of \emph{de Bruijn} strings, for all $w \in \bin^n, \occ{w}{x} = 1$. Hence, for all $v \in \bin^{n-1},\,  \occ{v0}{x} = \occ{v1}{x} = 1$. As $\occ{v}{x} \geq 2$ (as $\occ{0^{n-1}}{x} = 3$), a context for $v$ would have been created, and as $v$ is not unique in $x$, contexts for each its branching paths have to be created, namely $v0$ and $v1$. However, one more bit is required to finish building the context in the case where $v0 = 10^{n-1}$ (the last $n$ bits of $x$) as the model cannot build a context until it can say what it predicts. Hence $|x| + 1 = 2^n + n$ bits are needed in total.
 
\end{proof}
%The proof is contained in the appendix.

%
%
%
%
%
%
%
%
%

\subsection{The \emph{Bad Zone}}
\label{bad section}

For each $S_n$, its first $2^n + 2n$ bits are referred to as the \emph{bad zone}. Here we make no assumption about the contexts being used and assume worst possible compression. The hope is that after the  first $2^n + n$ bits of $S_n$ are encoded, either the contexts used to predict $S_n[0 .. 2^n + n -1]$ will have been deterministic and will continue to correctly predict the remaining bits of $S_n$, or that new deterministic contexts will have been created that correctly predict the remaining bits of $S_n$. Unfortunately this may not always occur in the succeeding $n$ bits. This commonly occurs if the original contexts used straddle two $S_i$ sections. 

For instance, consider $S_7[0..135] = db(7)\cdot0^71$ (by Remark \ref{db facts}). The final $1$ will be predicted by a context of length at least $7$ by Lemma \ref{context building}. $0^7$ is the context of length $7$ that may be used. However, $0^7$ is not a deterministic context. Since $S_4$ ends with $10^3$ and $S_5$ begins with $0^51$, this results in the substring $10^81$ straddling two sections. We have that $\occ{0^7}{10^81} =2$ with $\occ{0^70}{10^81} =1$ and $\occ{0^71}{10^81} =1$. Hence $10^7$ is a context of length $8$ that may be used. We know it exists as $\occ{10^6}{S \upharpoonright S_6 \cdot S_7[0..135]} = 2$. It occurs once in the straddle of $S_4$ and $S_5$ ($ 10^3 0^51$), and again in the straddle of $S_6$ and $S_7$ ($10 0^7 1$), both times followed by a $0$. It does not appear anywhere else (as $0^7$ cannot occur anywhere else) and so it deterministically predicts a $0$. However the bit currently being predicted is a $1$ and so an escape is required.

The following lemma puts an upper bound on how many bits are required to encode any singular bit occurring in an $S_n$ zone. The proof requires a counting argument examining how many times a context of length $n$ and $n-1$ can occur in the prefix $S_1S_2 \ldots S_n$ of $S$.

\begin{lemma}
\label{bad bit comp}
For almost every $n$, if $S \upharpoonright S_{n-1}$ has already been read, each bit in $S_n$ contributes at most $\log(n^5)$ bits to the encoding of $S$.
\end{lemma}
%The proof is contained in the appendix.
\begin{proof}[Proof: ]
For a fixed $n$, let $b$ be the current bit of $S_n$ being encoded. Let $x$ be the context used to predict $b$ by the encoder. Then $|x| \geq n-1$ as all contexts of length $n-2$ and below are non-deterministic in $S_n$ as seen in Lemma \ref{context building}. In the worst case scenario, $x$ will be deterministic, but will not predict $b$ correctly and thus will transmit an escape. In the worst case for $j \leq n-1$, $\occ{x}{S_j} \leq j$, i.e. once for every instance of $db(j)$ in $S_j$,  and $\occ{x}{S_n} = 2n$. Thus, we can bound the maximum possible number of occurrences by $$\sum_{j=1}^{n} j \leq n^2.$$ This results in an escape being transmitted in at most $-\log(\frac{1}{n^2 + 1})$ bits. 

The $b$ will then be transmitted by the non-deterministic context $x[1 .. |x|-1]$, of length at least $n-2$ (since $x$ is originally chosen as the shortest deterministic relevant context). Using the same logic, this context will have appeared at most $j$ times in $S_j$ for $j \leq n-2$, and $2(n-1)$ times in $S_{n-1}$ and $4n$ in $S_n$. Thus, we can bound the maximum possible number of occurrences by $$(\sum_{j=1}^{n-2} j) + 2(n-1) + 4n \leq n^2$$ for $n$ large. Hence, $b$ will be transmitted in at most $-\log(\frac{1}{n^2 + 2})$ bits. As such, $b$ contributes at  most $$\log(n^2 + 2) + \log(n^2 + 1) \leq \log(n^5)$$ bits to the encoding for $n$ large.
 
\end{proof}

The above Lemma \ref{bad bit comp} and knowing the size of the \emph{bad zone} allows us to bound the number of bits contributed by the bad zone of $S_n$.
\begin{corollary}
\label{bad zone comp}
For almost every $n$, if $S \upharpoonright S_{n-1}$ has already been read, the \emph{bad zone} of $S_n$ contributes at most $(2^n + 2n)\log(n^5)$ bits to the encoding of $S$.
\end{corollary}

\subsection{Main Result}
\label{main sec}

In this section we prove our main result that $$R_{\PPMun}(S)= 0.$$ This compression is achieved from the repetition of the \emph{de Bruijn} strings which lead to repeated use of deterministic contexts. When deterministic contexts are used, correct predictions are performed with probability $\frac{k}{k+1}$, for some $k \in \N$. Note that as $k$ increases, the number of bits contributed to the encoding $(-\log(\frac{k}{k+1}))$ approaches $0$.

The following shows that for $n$ large, whenever $n$ is odd or $n= 2^j$ for some $j$, the bits of $S_n$ not in the \emph{bad zone} will be predicted by deterministic contexts.

\begin{lemma}
For $n$ large, where $n$ is odd or $n = 2^j$ for some $j$, all bits not in the \emph{bad zone} of $S_n$ are correctly predicted by deterministic contexts.
\label{Out bad Zone}
\end{lemma}

\begin{proof}[Proof: ]
For $n$-odd, the %bit following the %2^n + 2n?% 
$2^n + 2n + 1^\thh$ bit in $S_n$ will always be a $1$ (if $n$ is a power of $2$, a similar argument holds but we look at the $2^n + 2n^\thh$ bit). This is because $db(n)\cdot db(n)[0.. 2n] = db(n)\cdot0^n10^{n-2}11$ by Remark \ref{db facts}. The context used to predict this $1$ will always be a suffix of the  context $010^{n-2}1$. This context exists as we have that  $\occ{010^{n-2}1}{S_n[0.. 2^n + 2n]} = 2$.

\begin{claim}
The context $010^{n-2}1$ deterministically predicts the $1$.
\end{claim}

First we show that $\occ{010^{n-2}1}{S \upharpoonright S_{n-1}} = 0$.
The only place the string $10^{n-2}$ occurs is in $S_{n-1}$ where it would be preceded by $1^{n-2},$ in $S_{n-2}$ where it would be preceded by $1^{n-3}$ and not a $0$, or along a straddle between two prior $S_i$'s for $i \leq n-1$, but again, it would be preceded by a $1$, and not a $0$. Hence, $010^{n-2}1$ first appears in $S_n[0 .. 2^n-1] = db(n)$ where it is followed by a $1$, and so is deterministic. This established the claim.

As the context $010^{n-2}1$ is deterministic, all extensions of this context (those of the form $010^{n-2}1y$ for the appropriate $y \in \fbins$) that are built while reading $S_n$, must be deterministic also. They remain deterministic throughout the reading of $S_n$ since due to the construction of $S_n$, any substring of $S_n$ of length at least $n$ is always followed by the same bit. Thus, every bit not in the \emph{bad zone} of $S_n$ is predicted by a deterministic context which is a suffix of an extension of the deterministic context $010^{n-2}1$.
 
\end{proof}
%The proof is contained in the appendix.

For $n$-even but not of the form $2^j$ for some $j$, Lemma \ref{Out bad Zone} does not hold as while most bits are predicted by deterministic contexts, the shifts of the \emph{de Bruijn} strings in the construction of $S_n$ between blocks $B_{n,0}$ and $B_{n,1}$ mean that some contexts which may have originally been deterministic in $B_{n,0}$, soon see the opposite bit due to these shifts at the start of $B_{n,1}$.

For instance, consider the string $1^60^5$. We have that $\occ{1^60^5}{S \upharpoonright S_5} = 0$ as $1^6$ is not contained in any \emph{de Bruijn} string of order less than $6$. However it does occur in $S_6$ multiple times. The first two times it occurs it sees a $0$ (as $db(6)[2^6 - 7 .. 2^6 -1]\cdot db(6)[0 .. 5] = 1^60^6)$ by Remark \ref{db facts}). However the third time it sees a $1$ due to the shift in $B_{6,1}$ (as $db(6)[2^6 - 7 .. 2^6 -1]\cdot db_1(6)[0 .. 5] = 1^60^51$). Hence, $1^60^5$ is no longer deterministic.

We first prove the following result which bounds the number of bits each block $S_n$ contributes to the encoding. In the following, $|\PPMun(S_n | S\upharpoonright S_{n-1})|$ represents the number of bits contributed to the output by the $\PPMun$ encoder on $S_n$ if it has already processed $S\upharpoonright S_{n-1}$.

\begin{theorem}
\label{upper bound}
For almost every $n$, $$|\PPMun(S_n \, | \, S \upharpoonright S_{n-1})| \leq (2^n + 2n + n^2)\log(n^5) + \log((n-1)^n(n)^{2^{n+1}}).$$
\end{theorem}

\begin{proof}[Proof: ]
By Lemma \ref{Out bad Zone}, every bit outside the \emph{bad zone} is predicted correctly by a deterministic context when $n$ is odd or when $n$ is a power or $2$. This is not true for the remaining $n$ as mentioned in the discussion preceding this theorem. As such, the output contributed by the case where $n$
 is even but not a power of $2$ acts as an upper bound for all $n$.
 
In this case, $n = 2^st$, for $s,t, \in \N,$ where $t$ is odd. Recall that $S_n = B_0 \cdot B_1  \cdots B_{2^s - 1}$ where $B_i = (db_{i}^t(n))$. Let $b_n = 2^s$, the number of blocks in $S_n$. Unlike in the other two cases where all contexts used to encode remained deterministic throughout the encoding of $S_n$ after the first $2^n + 2n$ bits are processed by Lemma \ref{Out bad Zone}, in this case some contexts do not due to the shifts that occur within each block. If there are $b_n$ blocks, there are $b_n-1$ shifts. However, we can pinpoint which bits are predicted by deterministic contexts.

After processing the \emph{bad zone}, a $1$ is deterministically correctly predicted by the context $010^{n-2}1$ with probability at least $\frac{1}{2}$. This is because $\occ{010^{n-2}1}{S \upharpoonright S_{n-1}} = 0$ since the only place the string $10^{n-2}$ occurs is in $S_{n-1}$ where it would be preceded by $1^{n-2}$ and not a $0$, in $S_{n-2}$ where it would be preceded by $1^{n-3}$ and not a $0$, or along a straddle between two prior $S_i$'s for $i \leq n-1$, where it would be preceded by a $1$ and not a $0$. Hence, $010^{n-2}1$ first appears in $S_n[0 .. 2^n-1] = db(n)$ where it is followed by a $1$ (by Remark \ref{db facts}), and so is deterministic.

Following this, the next $2^n - n - 2$ bits will also be predicted by a deterministic context successfully throughout the process. This is because these contexts are suffixes of extensions of $010^{n-2}1$ and see the same bits within $S_n$. The next time $010^{n-2}1$ is seen it predicts a $1$ with probability at least $\frac{2}{3}$ and so on. When $010^{n-2}1$ is seen for the $n^{\thh}$ time, there are only $2^n - 2n + b_n-1$ bits left to encode as the encoder has \emph{fallen behind} by $b_n-1$ bits due to the $b_n-1$ shifts that occur. These bits are encoded with probability at least $\frac{n-1}{n}$. Excluding the \emph{bad zone} we have accounted for $$(2^n - n - 1)(n-2) + (2^n - 2n + b_n -1) = 2^nn - 2^n - n^2 - n + 1 + b_n$$ bits. These are encoded in 
\begin{align*}
    -\log(((\frac{1}{2})^{2^n - n - 1}(\frac{2}{3})^{2^n - n - 1}&\cdots (\frac{n-2}{n-1})^{2^n - n - 1})(\frac{n-1}{n})^{2^n - 2n + b_n - 1}) \notag\\
 %=    -\log((\frac{1}{n-1})^{n-b_n}(\frac{1}{2})^{2^n - 2n + b_n - 1}) \\
 %&= \log((n-1)^{n-b_n}(n)^{2^n - 2n + b_n - 1}) \\
& \leq \log((n-1)^n(n)^{2^n + 2^n}) \tag{as $b_n < 2^n$ }  \\
& = \log((n-1)^n(n)^{2^{n+1}}) \tag{\dag}
\end{align*}
bits via arithmetic encoding.

Things differ with other contexts used as they may be impacted by the shifts that occur between blocks as discussed previously. For simplicity it is assumed all bits not accounted for so far contribute the worst case number of bits possible to the encoding. There are 
%$$2^nn - (2^n + 2n) - (2^nn - 2^n - n^2 - n + 1 + b_n) = n^2 - n - 1 - b_n$$
$n^2 - n - 1 - b_n$ such bits.

Then by \dag, Lemma \ref{bad bit comp} and Corollary \ref{bad zone comp}, we have that $$|\PPMun(S_n \,|\, S\upharpoonright S_{n-1})| \leq (2^n + 2n + n^2)(\log(n^5)) + \log((n-1)^n(n)^{2^{n+1}}).$$  
\end{proof}

We now prove the main theorem.

\begin{theorem}
\label{main thm}
$R_{\PPMun}(S) = 0.$
\end{theorem}
\begin{proof}[Proof: ]
Note that the worst compression of $S$ is achieved if the input ends with a complete bad zone, i.e. for a prefix of the form $S \upharpoonright m = S_1\ldots S_{n-1}S_{n}[0..2^n + 2n -1]$ for some $n$.

Let $S \upharpoonright m$ be such a prefix and let $k$ be such that Theorem \ref{upper bound} holds for all zones $S_i$ with $i \geq k$. The prefix $S_1\ldots S_{k-1}$ will always be encoded in O(1) bits. This gives

\begin{align*}
    \limsup\limits_{m \to \infty} \frac{|\PPMun(S \upharpoonright m)|}{m} & \leq \lim\limits_{n \to \infty}\Bigg (\frac{\sum\limits_{j = k}^{n-1}(2^j + 2j + j^2)\log(j^5)}{|S_1 \ldots S_{n-1}| + (2^{n} + 2n)} \Bigg )  \tag{by Thm \ref{upper bound}} \\
    &+ \lim\limits_{n \to \infty}\Bigg (\frac{(2^{n} + 2n)\log(n^5) + O(1)}{|S_1 \ldots S_{n-1}| + (2^{n} + 2n)} \Bigg )\\
    &= 0.
\end{align*}

As the overhead contributes at most one bit, we have that $R_{\PPMun(S)} = 0.$
 
\end{proof}
As $S$ is a normal sequence we have the following result.

\begin{corollary}
\label{normal S}
There exists a normal sequence $S$ such that $R_{\PPMun}(S) = 0.$
\end{corollary}

\subsection{Comparison of $\PPMk{k}$ and $\PPMun$}
\label{ppmk sec}

The following theorem demonstrates that for all $k \in  \N$, $\PPMk{k}$ achieves a best-case compression ratio of at least $\frac{1}{2}$ on $S$. Suppose you are examining $\PPMk{k}$. The idea is that each context of length $k$ predicts the same number of $0$s and $1$s in each $S_n$ zone for $n \geq k$. For $x \in \bin^k, n \geq k$, suppose $\occ{x}{S_n}=t$. The least amount any bit can contribute in $S_n$ is if the first $\frac{t}{2}$ times $x$ is seen it sees a $0$ and the final $\frac{t}{2}$ times it is seen it sees a $1$ (or vice versa). The $\frac{t}{2}^{\thh}$ $0$ (or $1$) contributes the least amount of bits, and if this amount is used as a lower bound for every bit in $S_n$, this gives the lower bound of $\frac{1}{2}$.  For each $k \in \N$ we use $\PPMk{k}(x)$ to denote the compression of $x \in \fbins$ when the max context length is bounded to be k.

\begin{theorem}
\label{ppmk comp S}
There exists a sequence $S$ such that $R_{\PPMun}(S) = 0$ but for all $k \in \N$, $\rho_{\PPMk{k}}(S) \geq \frac{1}{2}.$
\end{theorem}
%The proof is in the appendix.

\begin{proof}[Proof: ]

Let $S$ be our sequence from Theorem \ref{main thm} that $R_{\PPMun}(S) = 0$. Let $k$ be the maximum context length for the bounded PPM compressor \PPMk{k}. Recall $S = S_1S_2\ldots$ where $S_i$ is an enumeration of all strings of length $i$.

Once $\PPMk{k}$ processes $S_1S_2\ldots S_k$, the model will contain a context for every string of length $k$. Let $x \in \{0,1\}^k$. Let $n_{x,b}$ be the number of instances that $x$ has been followed by $b$ in $S_1 \ldots S_{k-1}$, for $b \in \{0,1\}$. This means that if $x$ is next followed by $0$, it will be predicted with probability $\frac{n_{x,0}}{n_{x,0}+n_{x,1}+2}$. Let $t > k$, and consider the substring $S_t'=S_tS_{t+1}[0\ldots k-1]$ of $S$. A context $x\in \{0,1\}^k$ will appear $t\cdot2^{t-k}$ times in $S_t'$, half the time followed by a $0$, and half the time followed by a $1$. The maximum compression of $S_t'$ that could be achieved is when each context prediction contributes as few bits as possible. The minimum amount that can possibly be contributed by any single prediction occurs if the first $t\cdot2^{t-k-1}$ times $x$ is seen it is always followed by the same bit $b$, and the remaining times by $\hat{b}$, that is $b$ flipped. This $t\cdot2^{{t-k-1}^{\thh}}$ instance of $x$ being followed by $b$ contributes the fewest amount of bits possible to the final encoding. Of course, this is a hypothetical scenario and does not actually occur in our $S$, it serves as a lower bound.

Suppose $b=0$. Then in this hypothetical scenario for $S_{t,0}'$, the probabilities that a $0$ is predicted is given by the sequence
\begin{align*}
S_{t,0}' =& \bigg\{\frac{n_{x,0} + \sum_{n=k+1}^{t-1}(n\cdot2^{n-k-1}) + j}{(n_{x,0}+n_{x,1}+2)+\sum_{n=k+1}^{t-1}(n\cdot2^{n-k}) + j}\bigg\}_{0\leq j < t\cdot2^{t-k}} \\
=&\bigg\{\frac{n_{x,0} + (1-k+t\cdot2^{t-k-1}-2^{t-k}) +j}{(n_{x,0}+n_{x,1}+2)+ (2-2k + t\cdot2^{t-k} - 2^{t-k+1}) + j}\bigg\}_{0\leq j < t\cdot2^{t-k}}.
\end{align*}

Note that 
$$
    \limsup\limits_{\substack{t \rightarrow \infty \\ 0 \leq j < t\cdot2^{t-k}}}S_{t,0}' = \frac{2}{3} \text{\, \, and} 
   \liminf\limits_{\substack{t \rightarrow \infty \\ 0 \leq j < t\cdot2^{t-k}}}S_{t,0}' = \frac{1}{2}.
$$

A sequence for $S_{t,1}$ can similarly be defined to get $$\limsup\limits_{\substack{t \rightarrow \infty \\ 0 \leq j < t\cdot2^{t-k}}}S_{t,1}' = \frac{1}{2} \text{\, \, and} 
    \liminf\limits_{\substack{t \rightarrow \infty \\ 0 \leq j < t\cdot2^{t-k}}}S_{t,1}' = \frac{1}{3}.
    $$
    
    Thus for small $\epsilon$, as $m$ gets large each new prediction contributes at least $-\log(\frac{2}{3}+\epsilon)$ bits. Hence, for $\delta > 0$, for almost every $m$ we have

$$|\PPMk{k}(S\upharpoonright m)| \geq m(1-\delta)(-\log(\frac{2}{3}+\epsilon)) \geq m(1-\delta)(-\log(\frac{7}{10})) \geq (1-\delta)\frac{m}{2}.$$
 
\end{proof}

The bound in the above theorem can of course be made much tighter, but it is sufficient to demonstrate a difference between $\PPMun$ and $\PPMk{k}$.

\section{Lempel-Ziv 78}

The Lempel-Ziv 78 (LZ) algorithm  \cite{b.lempel-ziv} is a lossless dictionary based compression algorithm. Given an input $x \in \fbins$, $\LZ$ parses $x$ into phrases $x = x_1x_2\ldots x_n$ such that each phrase $x_i$ is unique in the parsing, except for maybe the last phrase. Furthermore, for each phrase $x_i$, every prefix of $x_i$ also appears as a phrase in the parsing. That is, if $y$ is a prefix of $x_i$, then $y = x_j$ for some $j<i$. Each phrase is stored in $\LZ$'s dictionary. $\LZ$ encodes $x$ by encoding each phrase as a pointer to its dictionary containing the longest proper prefix of the phrase along with the final bit of the phrase. Specifically for each phrase $x_i$, $x_i = x_{l(i)}b_i$ for $l(i) < i$ and $b_i \in \{0,1\}.$ Then for $x = x_1x_2\ldots x_n$
$$\LZ(x) = c_{l(1)}b_1c_{l(2)}b_2\ldots c_{l(n)}b_n$$
where $c_i$ is a prefix free encoding of the pointer to the $i^{th}$ element of $\LZ$'s dictionary, and $x_0 = \lambda$.

Sequences that are enumerations of all strings are incompressible the LZ algorithm. As such, taking $S$ from Theorem \ref{main thm}, by Corollary \ref{normal S} and as $S$ is an enumeration of all strings, we have the following result.

\begin{theorem}
There exists a normal sequence $S$ such that \begin{enumerate}
    \item $R_{\PPMun}(S) = 0$,
    \item $\rho_{LZ}(S) = 1$.
\end{enumerate}
\label{Compare Theorem}
\end{theorem}

\section{Open Questions}
Does there exists a  sequence which acts as the opposite to Theorem \ref{Compare Theorem}?  Can the construction method for $S$ in Theorem \ref{main thm} be  generalised so that an infinite family of sequences satisfy the theorem? Bounded PPM and \PPMun gives rise to the possibility of developing a notion of Bennett's logical depth \cite{b:bennett88} based on the PPM algorithms. $S$ is an obvious candidate for a PPM-deep sequence, but how would properties such as the Slow Growth Law be defined in the PPM setting? Depth notions based on compressors and transducers have already been introduced in \cite{DBLP:conf/cie/DotyM07,jm2020difference}.
\bibliographystyle{plain}
\bibliography{main.bib}

\end{document}